\documentclass[12pt,a4paper]{article}

\usepackage{geometry}
\usepackage[utf8]{inputenc}
\usepackage[T1]{fontenc}
\usepackage{amsmath,amssymb,amsthm}
\usepackage{graphicx}
\usepackage{hyperref}
\usepackage{authblk}      
\usepackage{fancyhdr}     
\usepackage{titlesec}     
\usepackage{natbib}
\makeatletter
\renewcommand\@biblabel[1]{}
\makeatother
\titleformat{\section}
  {\normalfont\normalsize\bfseries}{\thesection.}{0.5em}{}
\titleformat{\subsection}
  {\normalfont\normalsize}{\thesubsection.}{0.5em}{}
\titleformat{\subsubsection}
  {\normalfont\normalsize\itshape}{\thesubsubsection.}{0.5em}{}
  
\titlespacing{\section}{0pt}{1.5ex plus 0.5ex minus 0.2ex}{0.8ex plus 0.2ex}
\titlespacing{\subsection}{0pt}{1.2ex plus 0.4ex minus 0.2ex}{0.6ex plus 0.2ex}
\titlespacing{\subsubsection}{0pt}{1ex plus 0.3ex minus 0.1ex}{0.5ex plus 0.1ex}

\fancypagestyle{firstpage}{
    \fancyhf{}
    \fancyfoot[C]{\thepage}
    
}

\newcommand{\R}{\mathbb{R}}
\newcommand{\mP}{\mathcal{P}}
\def\hsig{\hat{\Sigma}}
\newcommand{\n}{\frac{1}{n}\sum_{i=1}^n}
\def\sE{\mathcal{E}}

\newcommand{\bigO}{\mathcal{O}}

\newcommand{\argmin}{\mathop{\mathrm{arg\,min}}}

\newcommand{\tr}{\mathrm{tr}}
\newcommand{\tran}{\mathrm{T}}
\def\Cov{\mathop{\mathrm{Cov}}\nolimits}
\newcommand{\ave}{\mathrm{ave}}
\newcommand{\MM}{\mathcal{M}}
\newcommand{\mA}{\mathcal{A}}

\newtheorem{theorem}{Theorem}
\newtheorem{lemma}{Lemma}
\newtheorem{example}{Example}

\title{\fontsize{16}{20}\selectfont\textbf{Robust regularized covariance matrix estimation: well-posedness and convergent algorithm}}

\author[1]{Mengxi Yi}
\author[2]{David E. Tyler}

\affil[1]{School of Statistics, Beijing Normal University, Beijing, China}
\affil[2]{Department of Statistics, Rutgers University, Piscataway, USA}

\date{}  

\newcommand{\keywords}[1]{\par\noindent\textbf{Keywords:} #1}

\begin{document}

\maketitle
\thispagestyle{firstpage}

\begin{abstract}
\noindent
In this paper, we study properties of penalized and structured M-estimators of multivariate scatter, based on geodesically convex but 
not necessarily smooth penalty functions. Existence and uniqueness conditions for these penalized and structured estimators are given. However, we show that the standard fixed-point algorithm which is usually applied to an M-estimation problem does not necessarily converge for penalized M-estimation problems. Hence, we develop a new but simple re-weighting algorithm and prove that it has monotone convergence for a broad class of penalized and structured M-estimators of multivariate scatter.

\vspace{1em}
\keywords{Constrained estimation, Geodesic convexity, M-estimation, MM-algorithm, Regularization, Reweighting algorithm.}
\end{abstract}

\vspace{2em}

%

\section{Introduction and Motivation} \label{Sec:M-intro}
Covariance matrix estimation is an important topic in multivariate statistics. For a random sample $x_1,\cdots,x_n\in\R^p$, the classical estimate of the covariance matrix $\Sigma=\mathbb{E}[(x-\mathbb{E}x)(x-\mathbb{E}x)^{\tran}]$ is the empirical covariance matrix, $S_n=\frac{1}{n}\sum_{i=1}^n(x_i-\bar{x})(x_i-\bar{x})^{\tran}$, where $\bar{x}=\frac{1}{n}\sum_{i=1}^nx_i$. When the sample size $n$ is small relative to the dimension $p$, and in particular when $n < p$,  $S_n$ is well-known to be a rather poor estimate of $\Sigma$ even when sampling from a multivariate normal distribution.
In such cases, one may wish to consider more parsimonious models for the $p(p+1)/2$ parameters in the population covariance matrix,
such as a Toeplitz model \cite{Snyder:1989}, factor model \cite{Fan:2008} or a graphical model \cite{Dempster:1972}. Alternatively, one may wish to use a 
regularized or a penalized version of the sample covariance matrix, such as the Ledoit-Wolf estimator \cite{Ledoit-Wolf:2004} or the graphical lasso \cite{Yuan-Lin:2007}, among others \cite{Warton:2008, Deng-Tsui:2013}. Since the negative log likelihood, taken as a loss function, under multivariate normal sampling is convex in the inverse of the covariance matrix, $\Sigma^{-1}$, it is natural to consider additive penalties which are also convex in the inverse of the covariance matrix. More recently, \cite{Tyler-Yi:2018} and \cite{Tyler-Yi:2019} studied a class of smooth and non-smooth penalties which do not have this convexity property, but rather have the property of being geodesically convex. By utilizing the little known property that the negative log likelihood under multivariate normal sampling is also geodesically convex (g-convex), they show that such penalized sample covariance problems yield unique solutions.

The sample covariance matrix is well known to be highly non-resistant to outliers in the the observations and to be a highly inefficient estimator when the samples are drawn from heavy-tailed distributions. Robust estimators can protect against this. Here we concentrate on robust penalized method for estimating covariance matrices for multidimensional observations. In the context of robust estimation, the covariance matrix is often referred to as the pseudo-covariance or scatter matrix to allow for nonexistence for the second order moments. Our aim is to study robust versions of the penalized sample covariances, which are defined by applying a rigorous and general treatment of penalization to loss functions which are more robust than the normal negative log-likelihood function. In particular, we consider the class of loss functions which give rise to the monotonic M-estimators of multivariate scatter. Using the concept of geodesic convexity, we derive necessary and sufficient conditions ensuring the well-posedness of a broad class of penalized M-estimators of scatter. These well-posedness conditions, which guarantee both existence and uniqueness of the estimator, are generally weaker than the traditional conditions required for M-estimator of scatter. 

Since its introduction, a proven approach for finding a solution to M-estimating equations has been via fixed-point algorithm, which can be related to iterative re-weighted least squares algorithms \cite{Tyler:1997, Kent-Tyler:1991}. Although it is known to work in the unpenalized setting, in general this algorithm is not necessarily applicable to the penalized setting. Furthermore, even though a fixed point algorithm has been successfully applied when using the Kullback-Leibler penalty \cite{Ollila-Tyler:2014}, there are no convergence results for general penalties. One result we show in this paper is that this algorithm does not necessarily converge for other penalty functions, and in particular for the Riemannian penalty. Consequently, we propose in section \ref{Sec:Alg} a new but simple re-weighting algorithm. We prove this algorithm always converges for any g-convex penalized or structured M-estimator of multivariate scatter, and that the convergence is monotone.

The rest of the paper is organized as follows. Section \ref{Sec:M-est} reviews the M-estimates of multivariate scatter, introduces the general framework for the penalized version and establishes some new results regarding their existence and uniqueness. In addition, subsection \ref{Sec:Bayes} discusses a Bayesian interpretation of penalized M-estimation. In section \ref{Sec:Alg}, the aforementioned re-weighting algorithm for the penalized  case is presented and some examples are given therein. Section \ref{Sec:constraint} demonstrates the duality between the problem of minimizing a g-convex loss function and a minimization problem under a g-convex constraint. Algorithms for a general g-convex constraint problem are also shown here. Section \ref{Sec:con} includes a concluding remark. The concept of g-convexity is reviewed in Appendix \ref{App:gc} and proofs are given in Appendix \ref{App:proofs}.

\section{Regularized M-estimators of multivariate scatter} \label{Sec:M-est}
\subsection{Review of elliptical distributions and the M-estimators of scatter}

A random vector $x\in\R^p$ is said to have a real \emph{elliptically symmetric distribution} with center $\mu$ and scatter matrix $\Sigma\in\mP^p$, denoted by $\sE(\mu, \Sigma, g)$, if it has a density of the form
\[
f(x)=C_{p,g}\det(\Sigma)^{-1/2}g\{(x-\mu)^{\tran}\Sigma^{-1}(x-\mu)\},
\]
where $\mP^p$ represents class of positive definite matrices of order $p$, $C_{p,g}>0$ denotes a normalizing constant ensuring that $f(x)$ integrates to one, and $g$ is a given function $g:\R^{+}\to\R^{+}$, where $\R^{+}=\{x\in\R|x\ge 0\}$. The scatter matrix $\Sigma$ is proportional to the covariance matrix whenever $x$ possesses second moments. When the second moments do not exist, the scatter matrix $\Sigma$ can be viewed  as a generalization of the covariance matrix. Note that over the semi-parametric class of elliptical distributions, i.e.\ when the radial function $g$ is not specified, $\Sigma$ is only well-defined up to proportionality.  The function $g$ determines the radial distribution of the elliptical population and hence the degree of its ``heavy-tailedness". For example, when $g(s)=\exp(-s/2)$, one obtains the multivariate normal distribution. If $x$ is a random variable having an elliptical distribution $\sE(\mu, \Sigma, g)$, then the standardized variable $z=\Sigma^{-1/2}(x-\mu)$ has a \emph{spherical distribution} centered at $0$, with density $C_{p,g}g(z^{\tran}z)$. The spherical $z$ admits a decomposition $z=rs$, where $r=||z||$ and $s=z/||z||$ are independent, with $s$ being uniformly distributed on the unit sphere. For a detailed account on properties of elliptical distributions, we refer readers to \cite{Kelker:1970,Ollila:2012}.

Since this paper focuses on the covariance matrix and its generalizations, we presume throughout that the $p$-dimensional sample $\{x_1, \ldots, x_n\}$
has been robustly centered by an estimate of $\mu$, e.g. by the vector of marginal medians or the spatial median, or by a known center. Alternatively, the data may be replaced with their pairwise differences. Hereafter, we assume $\mu=0$.

For a given function $g$, the maximum likelihood estimator (MLE) of the scatter matrix, is then the minimum of the corresponding negative log-likelihood function
\begin{equation} \label{eq:ml}  
\ell_{\rho}(\Sigma) = \frac{1}{n} \sum_{i=1}^n \rho(x_i^{\tran} \Sigma^{-1} x_i) + \log\{ \det(\Sigma)\}, 
\end{equation}
over all $\Sigma > 0$, i.e. over all positive definite symmetric matrices of order $p$, here $\rho(s)=-2\log g(s)$. If $\rho$ is differentiable, then setting the derivative of $\ell_{\rho}(\Sigma)$, with respect to $\Sigma$, to zero yields the estimating equation:
\begin{equation}\label{eq:mequation}
\Sigma = \n u(x_i^{\tran}\Sigma^{-1} x_i) x_i x_i^{\tran},
\end{equation}
where $u(s) = \rho^\prime(s)=-2g^{\prime}(s)/g(s)$ represents a weighting function. The critical points of the objective function \eqref{eq:ml} then satisfy the estimating equation \eqref{eq:mequation}.

M-estimators of scatter matrices were defined in \cite{Maronna:1976} as generalizations of the MLEs of an elliptical distribution, namely by allowing for a general $\rho$ function in (\ref{eq:ml}), which need not be related to an elliptical distribution. Hence we refer to (\ref{eq:ml}) as an M-loss function in general, while equation (\ref{eq:mequation}) is called an M-estimating equation. Some examples of ML- and M-estimators are given below.

\emph{Sample covariance matrix.} For the Gaussian sample, $\rho(s)=s$, and $u(s)=\rho^{\prime}(s)=1$, and so (\ref{eq:ml}) becomes
\begin{equation}\label{eq:norm}
\ell(\Sigma;S_n)=\tr(\Sigma^{-1}S_n)+\log\{\det(\Sigma)\}
\end{equation}
where $S_n=\n x_ix_i^{\tran}$ is the sample covariance matrix. When $n> p$, $S_n$ is the MLE and an unbiased estimator of $\Sigma$.

\emph{MLE for an elliptical $t$-family.} For the $t_{\nu}$-distribution, $g(s) = (1+s/\nu)^{-(\nu + p)/2}$ and so $\rho(s)=(\nu+p)\log(\nu+s)$ and $u(s)=(\nu+p)/(\nu+s)$, with $\nu$ being the degrees of freedom. A $t_{\nu}$-distribution has finite second moments $\Cov(x)=\nu/(\nu-2)\Sigma$ only when $\nu>2$. When $\nu=1$, the corresponding distribution is said to have a $p$-variate elliptical Cauchy distribution. The limiting case $\nu\to\infty$ yields the normal distribution.

The loss function (\ref{eq:ml}) associated with the elliptical $t_{\nu}$ distribution is
\[
\ell^t(\Sigma)=\frac{\nu+p}{n}\sum_{i=1}^n\log(\nu+x_i^{\tran}\Sigma^{-1}x_i)+\log\{\det(\Sigma)\}.
\]
The existence and uniqueness of the $t_{\nu}$-MLEs have been considered in \cite{Kent-Tyler:1991}. When $\nu$ goes to zero, the $t_{\nu}$-MLE yields \emph{Tyler's M-estimator} \cite{Tyler:1987a}, which is the unique solution up to proportionality to the M-estimating equations:
\begin{equation*}\label{eq:Tyler}
\Sigma=\frac{p}{n}\sum_{i=1}^n\frac{x_ix_i^{\tran}}{x_i^{\tran}\Sigma^{-1}x_i}.
\end{equation*}
Tyler's M-estimator also corresponds to the MLE of $\Sigma$ when the data comes from an \emph{angular central Gaussian distribution}, for which $\rho(s)=p\log s$ and $u(s)=p/s$ \cite{Tyler:1987b}. The corresponding loss function can then be written as
\[
\ell^{Tyler}(\Sigma)=\frac{p}{n}\sum_{i=1}^n\log(x_i^{\tran}\Sigma^{-1}x_i)+\log\{\det(\Sigma)\}.
\]
For more details and properties of this estimator, see \cite{Tyler:1987a,Tyler:1987b}.


\emph{Huber's M-estimator.} Huber's M-estimator \cite{Huber:1964, Huber:1977} is defined via weight
$$u(s)=\begin{cases}
1/b, & \text{for } s\le c^2\\
c^2/(sb), & \text{for } s>c^2\end{cases}$$
where $c$ is a tuning constant defined so that $r=F_{\chi^2_{p}}(c^2)$ for a chosen $r\,(0< r\le1)$, where $F_{\chi^2_{p}}(\cdot)$ denotes the c.d.f. of the $\chi^2_{p}$. The scaling factor $b$ is usually chosen so that the resulting M-estimate is consistent for the covariance matrix at the normal distribution. When $r=1$,  $u(s)=1$, which corresponds to the sample covariance matrix, whereas $r \to 0$ yields $u(s)=p/s$, corresponding to Tyler's M-estimator.

\subsection{Penalized M-estimation of Scatter}
An obvious way to ``robustify'' the penalized sample covariance matrix is to replace $S_n$  with an M-estimator or some
other robust estimators of covariance. This may be a good approach for large sample sizes, but for smaller sample sizes properties of the M-estimators 
and other affine equivariant estimators of covariance matrices tend not to differ substantially from the sample covariance matrix; see \cite{Tyler:2010}. 
Given that regularization or penalized approaches are primarily of interest when the problems are ill-posed, such alternative methods are
worth exploring. In particular, penalized M-estimators of scatter have been proposed \cite{Warton:2008, Ollila-Tyler:2014}, which are defined as a minimizer over $\Sigma > 0$ of 
\begin{equation} \label{eq:mL} 
L_{\rho}(\Sigma; \eta) = \ell_{\rho}(\Sigma) + \eta \Pi(\Sigma), \end{equation}
where $\Pi(\Sigma)$ denotes a non-negative penalty function, with $\eta \ge 0$ being a tuning parameter. In particular, the penalized sample covariance is defined as a minimizer of the penalized normal likelihood function, which is written as
\begin{equation} \label{eq:nL}
    L(\Sigma;S_n;\eta)=\tr(\Sigma^{-1}S_n)+\log\{\det(\Sigma)\}+\eta\Pi(\Sigma).
\end{equation}

A detailed study of the penalized normal problem based on general penalties are given in \cite{Tyler-Yi:2018,Tyler-Yi:2019}. These studies include proven convergent algorithms for finding the minima of \eqref{eq:nL}. However, finding a minimizer to \eqref{eq:mL} is more involved than in the penalized sample covariance case \eqref{eq:nL}. It has been known since their introduction that, under certain conditions on the function $\rho(s)$ and on the sample, there exists a unique critical point to $\ell_{\rho}(\Sigma)$; see \cite{Maronna:1976, Kent-Tyler:1991} for details.  
This suggests $\ell_{\rho}(\Sigma)$ may possess some convexity property, but unlike the normal case it is not necessarily convex in $\Sigma^{-1}$. Some insightful work \cite{Zhang:2013}, though, has shown $\ell_{\rho}(\Sigma)$ to be geodesically convex in $\Sigma$
whenever the function $\rho(s)$ is convex in $\log(s)$, i.e.\ whenever $h: \R \to \R$ defined by $h(a) \equiv \rho(e^a)$ is convex. The notion of geodesic convexity, is reviewed in the appendix. As with convexity, any local minimum of a g-convex function is a global minimum, and when differentiable any critical point is a global minimum, with the set of all minima being g-convex. In addition, 
if a minimum exists, then the minimum is unique when the function is strictly g-convex. Finally, the sum of two g-convex functions is g-convex, and the sum is strictly g-convex if either of the two g-convex summands is strictly g-convex. These basic properties of geodesic convexity can be applied to establish the following lemma. 

\begin{lemma} \label{Lem:sumgeo}
If $\rho(s)$ and $\Pi(\Sigma)$ are g-convex, then $L_{\rho}(\Sigma; \eta)$ is g-convex and the solution set
\begin{equation} \label{eq:sol} 
\mA_{\rho} = \{\Sigma_* > 0 ~|~ L_{\rho}(\Sigma_*; \eta) \le L_\rho(\Sigma; \eta) \ \mbox{for all} \ \Sigma > 0 \} 
\end{equation}
is g-convex. Furthermore, if either $\rho(s)$ or $\Pi(\Sigma)$ is strictly g-convex, then $L_\rho(\Sigma; \eta)$ is strictly g-convex and hence
the solution set $\mA_{\rho}$ is either empty or contains a single element.
\end{lemma}

The above lemma applies when using the Kullback-Leibler penalty $\Pi_{KL}=\tr(\Sigma^{-1})+\log\{\det(\Sigma)\}$ or the Riemannian penalty  $\Pi_R(\Sigma)=||\log(\Sigma)||^2_F$, since in addition to being strictly convex in $\Sigma^{-1}$ and $\log(\Sigma)$ respectively, they are known to be strictly g-convex; see \cite{Tyler-Yi:2019} for details. 

As in the unpenalized case, the existence of a minimizer may depend on the data itself.  Other than being non-empty,
it is desirable for the solution set $\mA_{\rho}$ to be compact. The next lemma establishes that if these properties hold in
the well studied non-penalized case, $\eta = 0$, then they hold in the penalized case for any $\eta > 0$.
\begin{lemma} \label{Lem:Lexist}
Suppose $\rho(s)$ and $\Pi(\Sigma)$ are g-convex. If $\mA_{\rho}$ is a non-empty compact set when $\eta = \eta_o$, then it is a non-empty compact set when $\eta > \eta_o$. Furthermore, if either $\rho(s)$ or $\Pi(\Sigma)$ is strictly g-convex, the unique minimizer $\hsig_{\eta}>0$ is a continuous function of $\eta>\eta_o$.
\end{lemma}

A sufficient condition to guarantee $\mA_{\rho}$ is a non-empty compact set when $\eta = 0$ is
\begin{equation} \label{eq:C}
 P_n(\mathcal{V}) < 1-\frac{\{p - \mbox{dim}(\mathcal{V})\}}{K_\rho} \ \mbox{for any proper subspace} \ \mathcal{V} \subset \R^p, 
\end{equation}
where $K_{\rho} = \sup\{ k > 0 ~|~ s^k e^{-\rho(s)} \to 0 \ \mbox{as} \ s \to \infty \} $ and $P_n$ represents the empirical
probability measure for the data set, see \cite{Kent-Tyler:1991}. The constant $K_{\rho}$ is referred to as the \emph{sill} of $\rho(s)$. For $K_{\rho} > p$ and $n > p$, this condition holds with probability one 
when random sampling from a continuous multivariate distribution. When differentiable, $\rho(s)$ is g-convex if and only if the function $\psi(s) = su(s)$ is non-decreasing. In this case, $K_{\rho} = \psi(\infty)$. The function $\psi(s)$ is related to the influence function of
the M-estimator, with $K_{\rho}$ being related to its gross error sensitivity \cite{Hampel:1968, Hampel:1974}. Smaller values of $K_{\rho}$ imply smaller gross error sensitivities and hence Condition $\ref{eq:C}$ is more
restrictive for more robust M-estimators.

When $\eta > 0$, condition \eqref{eq:C} can usually be relaxed. For example,  \cite{Sun:2014} studied a family of shrinkage Tyler's estimators and provided sufficient conditions, dependent on the sample, for the existence and uniqueness of the estimators when using a Kullback-Leibler-type penalty. For a general g-convex $\rho$-function that is bounded from below, \cite{Duembgen-Tyler:2016}  stated that when using some g-coercive penalty functions, like the Riemannian penalty $\Pi_R$ or the symmetrized Kullback-Leibler penalty $\Pi_{KLs}(\Sigma)=\tr(\Sigma)+\tr(\Sigma^{-1})$, no conditions on the sample are needed when $\eta > 0$ to ensure a unique minimizer exists. We show in the next theorem that in general if the penalty is g-coercive and under some other conditions, independent of the data, a unique minimum always exists for large enough $\eta$. This result holds even for $n=0$.

\begin{theorem}\label{thrm:largeeta0}
Suppose $\rho(s)$ is bounded below, as well as g-convex, and $\Pi(\Sigma)$ is g-convex and g-coercive. If $\log\{\det(\Sigma)\}/\Pi(\Sigma)\to 0$ as $\log\{\det(\Sigma)\}\to-\infty$, the solution set $\mathcal{A}_{\rho}$ is non-empty for $\eta>0$. If $\log\{\det(\Sigma)\}/\Pi(\Sigma)$ is bounded below by $-\eta_0$ as $\log\{\det(\Sigma)\}\to-\infty$, 
the solution set $\mathcal{A}_{\rho}$ is non-empty for $\eta>\eta_0$.
\end{theorem}

G-coercivity for $\Pi(\Sigma)$ may not hold for some penalties of interest. In particular, it cannot hold for any shape penalty, i.e. a scale invariant penalty for which $\Pi(\Sigma)=\Pi(c\Sigma)$ for any $c>0$. As noted within the proof of Theorem \ref{thrm:largeeta0}, the condition that $\log\{\det(\Sigma)\}+\eta\Pi(\Sigma)$ be g-coercive is sufficient to guarantee the existence of a minimizer. This implies the following more general result.


\begin{theorem} \label{thrm:largeeta} 
Suppose $\rho(s)$ is bounded below, as well as g-convex, and $\Pi(\Sigma)$ is g-convex. If $\log\{\det(\Sigma)\}+\eta\Pi(\Sigma)$ is g-coercive for some $\eta=\eta_o$,
then the solution set $\mathcal{A}_{\rho}$ is non-empty for any $\eta\ge\eta_0$.
\end{theorem}

The conditions of Theorem \ref{thrm:largeeta0} hold for the Kullback-Leibler penalty $\Pi_{KL}$, the symmetrized Kullback-Leibler penalty $\Pi_{KLs}$, and the Riemannian penalty $\Pi_R$, with $\eta_o=0$. They do not hold for the inverse precision matrix penalty $\Pi_o(\Sigma)=\tr(\Sigma^{-1})$, considered in \cite{Ollila-Tyler:2014}. The conditions of Theorem \ref{thrm:largeeta}, though, do hold for $\Pi_o(\Sigma)$ for any $\eta>0$.

The conditions of Theorem \ref{thrm:largeeta} are still too strong for shape penalties. This can be seen by considering the sequence $\Sigma_k=c_kI_p$, with $c_k\to 0$, where $I_p$ is the identity matrix of order $p$. Here $\log\{\det(\Sigma_k)\}+\eta\Pi(\Sigma_k)=p\log c_k+\eta\Pi(I_p)\to -\infty$ for any $\eta\ge 0$. For shape penalties, the following modification of the above theorem can be used to establish existence of the corresponding penalized M-estimators of scatter.

\begin{theorem}\label{thrm:generalexist}
Suppose $\rho(s)$ is bounded below, as well as g-convex, and $\Pi(\Sigma)$ is g-convex with $\Pi(\Sigma)=\Pi(c\Sigma)$ for any $c>0$. Also, for any sequence $\Gamma>0$, with $\gamma_1=1$ and $\gamma_p\to 0$, where $\gamma_1\ge\cdots\ge\gamma_p>0$ are the eigenvalues of $\Gamma$. Suppose $\log\{\det(\Gamma)\}+\eta\Pi(\Gamma)\to\infty$ for some $\eta=\eta_o$. Then, provided $n\ge 1, K_{\rho}>p$ and $P_n(0)<1-p/K_{\rho}$, the solution set $\mA_{\rho}$ is non-empty for any $\eta\ge\eta_o$.
\end{theorem}

Examples of shape penalties which satisfy Theorem \ref{thrm:generalexist} are the Riemannian shape penalty $\Pi_{Rs}(\Sigma)=||\log(\Sigma/\det(\Sigma)^{1/p})||^2_F$, with the theorem holding for any $\eta>0$, and the non-smooth ``elasso" penalty $\Pi_E(\Sigma)=\sum_{i=1}^pa_i\log\lambda_i$, studied in \cite{Tyler-Yi:2018}, for  $\eta>\max\{-1/a_j, j=1,\cdots,p\}$, where $\lambda_1\ge\cdots\ge\lambda_p$ are eigenvalues of $\Sigma$ and $a_1\ge\cdots\ge a_p$ are fixed constants.

If we replace $\frac{1}{n}\sum_{i=1}^n\rho(x_i^{\tran}\Sigma^{-1}x_i)$ in \eqref{eq:mL} with $\frac{1}{n_1}\sum_{x_i\neq 0}\rho(x_i^{\tran}\Sigma^{-1}x_i)$, where $n_1=\#\{x_i\neq 0\}$, then Theorem \ref{thrm:generalexist} still holds if we replace the condition on $P_n(0)$ with the condition that $x_i\neq 0$ for some $x_i$. This is reasonable when we are only interested in the shape component, i.e.\ in functions of $\Sigma$ such that $H(\Sigma) = H(c\Sigma)$ for all $c > 0$, since $x_i=0$ has no information regarding shape. 

The above theorems do not hold for the regularized Tyler's M-estimators since $\rho(s)=plog(s)$ is not bounded from below. For this case, we can apply the following theorem.

\begin{theorem}\label{thrm:tylerexist}
Suppose $\rho(s)=plog(s)$, and $\Pi(\Sigma)$ is g-convex. Also, for any sequence $\Gamma>0$, with $\gamma_1=1$ and $\gamma_p\to 0$, where $\gamma_1\ge\cdots\ge\gamma_p>0$ are the eigenvalues of $\Gamma$. Suppose $\log\{\det(\Gamma)\}+\eta\Pi(\Sigma)\to\infty$ for some $\eta=\eta_o$. Then, the solution set $\mA_{\rho}$ is non-empty for any $\eta\ge\eta_o$.
\end{theorem}

\subsection{Bayesian interpretation}\label{Sec:Bayes}

 In regression, adding an $L_2$ penalty is equivalent to using a Gaussian prior on $\beta$, while adding an $L_1$ penalty is equivalent to using a Laplacian prior, i.e.
$$\beta\sim C\exp^{-\eta\beta^2} \quad \mbox{and} \quad \beta\sim C\exp^{-\eta|\beta|} \quad respectively.$$
Thus the ridge and lasso estimators of $\beta$ \cite{Tibshirani:1996} can be interpreted as the Bayesian posterior mode under these respective priors; see \cite{Li-Goel:2006} for a comprehensive review. Similarly, the penalized objective function \eqref{eq:mL} can be related to the
Bayesian posterior
$$\det(\Sigma)^{-n/2}\prod_{i=1}^ng(x_i^{\tran}\Sigma^{-1}x_i)\cdot \exp(-\eta(\Pi(\Sigma))).$$
When $\Pi(\Sigma)$ is orthogonally invariant, the prior, $\exp\{-\eta\Pi(\Sigma)\}$, corresponds to independent priors for $P$, and $\lambda$, where $\Sigma=P\Lambda P^{\tran}$ represents its spectral value decomposition, with $P$ having a Haar measure on $\bigO(p)$, and $\lambda \in \R^p$, which are the diagonal elements of $\Lambda$, having a density proportional to 
$\exp\{-\eta\pi(\log(\lambda))\}$. 

Since $L_{\rho}$ is g-convex, the corresponding posterior marginal density of $\Sigma$, say $f(\Sigma)$, is also g-convex. Thus its superlevel sets $\{\Sigma \in \mathcal{P}^p \,|\, f(\Sigma)\ge t\}$ are g-convex as well for any $t > 0$. As defined in \cite{Anderson:1955}, a multivariate function $h(x)$ is \emph{unimodal} if its epigraph $\{x\in\R^n\,|\,h(x)\ge c\}$ is convex for every real number $c$. This notion can be generalized to the space of positive definite symmetric matrix of order $p$. That is, the condition of unimodality of a matrix distribution function can be expressed by the condition that its superlevel sets are g-convex. Therefore, the posterior marginal density of $\Sigma$ is unimodal if the density generator function $g$ is monotonically non-increasing. Examples of unimodal elliptical density include normal, $t$-distribution, angular Gaussian distribution, and logistic distribution. 

Analogous to \cite{Anderson:1955} and \cite{Khatri:1967}, this generalization has applications to the concentration of matrix-variate distributions and the coverage probability of simultaneous confidence intervals for the components of the covariance matrix of a matrix-variate population. For example, if $X\sim N(0, \Sigma_1)$ and $Y\sim N(0, \Sigma_2)$ with $\Sigma_2-\Sigma_1\ge 0$, where $X, Y\in\R^p$. Then the distribution of $X$ is more concentrated about $0$ than $Y$, in the sense that$$P[Y\in \mathcal{J}]\le P[X\in \mathcal{J}]$$for every symmetric g-convex set $\mathcal{J}\in\R^p$.

\section{Reweighting algorithms}\label{Sec:Alg}
When $\rho(s)$ and $\Pi(\Sigma)$ are both differentiable then the minima of $L_{\rho}(\Sigma;\eta)$ satisfy the M-estimating equation:
\begin{equation*} \label{eq:meq}
\Sigma=\frac{1}{n}\sum_{i=1}^n\{u(x_i^{\tran}\Sigma^{-1}x_i)x_ix_i^{\tran}\}+\eta\nabla\Pi(\Sigma^{-1}),
\end{equation*}
where $u(s) = \rho^{\prime}(s)$ and $\nabla\Pi(\Sigma^{-1})$ denotes the derivative with respect to $\Sigma^{-1}$. A common way to find a solution to an M-estimating
equation is to then use a fixed-point (FP) algorithm:
\begin{equation}\label{eq:mq}
\Sigma_{k+1}=\frac{1}{n}\sum_{i=1}^n\{u(x_i^{\tran}\Sigma^{-1}_k x_i)x_ix_i^{\tran}\}+\eta\nabla\Pi(\Sigma_k^{-1}).
\end{equation}

For the unpenalized case, i.e.\ $\eta=0$, and when equation (\ref{eq:mequation}) corresponds to the MLE from a scale mixture of normals, such as a multivariate $t_{\nu}$-distribution, the FP algorithm arises from an application of a standard EM algorithm, properties of which have been well studied. In particular, the FP algorithm is thus monotonically decreasing in $\ell_{\rho}(\Sigma)$, meaning  $\ell_{\rho}(\Sigma_{k+1}) \le \ell_{\rho}(\Sigma_k)$. For \emph{monotonic} M-estimates,  i.e.\ when $\psi(s) = su(s)$ is non-decreasing, when $\rho(s)$ is unbounded and not necessarily related to a scale mixture of normals, \cite{Kent-Tyler:1991} show that the FP algorithm always converges to the unique solution, no matter the choice of the initial value. These conditions are satisfied, for example, by Huber's M-estimators of scatter, even though they do not correspond to MLEs based on a scale mixture of normals. Furthermore, it was later shown in  \cite{Arslan:2004} that for this more general case the FP algorithm is monotonically increasing in the corresponding M-objective function.

For the penalized case, \cite{Ollila-Tyler:2014} gives a partial proof that the fixed-point algorithm converges when using the Kullback-Leibler penalty or the trace precision matrix penalty $\Pi_o(\Sigma)=\tr(\Sigma^{-1})$. However, there is no general proof for other penalties. The following example shows that this algorithm doesn't always converge for the Riemannian penalty. 

\begin{example}\label{ex:1}
We consider a sample $x=\{x_1, \cdots, x_{100}\}\in\R^2$ generated from a mean-zero elliptical $t_3$-distribution, with $\Sigma$ being the identity matrix. We use the weight function generated from the $t_3$ distribution, i.e., $u(s_i)=\frac{\nu+p}{\nu+s_i}$ in equation (\ref{eq:mq}), where $s_i=x_i^{\tran}\Sigma^{-1}_kx_i, \nu=3$ and $p=2$.
 
First, consider the penalty $\Pi_o(\Sigma)=\tr(\Sigma^{-1})$. The fixed-point algorithm is given by
\[
\Sigma_{k+1}\leftarrow \frac{\nu+p}{n}\sum_{i=1}^n\frac{x_ix_i^{\tran}}{\nu+x_i^{\tran}\Sigma^{-1}_kx_i}+2\eta I_p
\]
For a fixed $\eta=0.5$, we compute $\hat{\Sigma}_{\eta}$ with  two different initial values given, respectively, by
\[
\Sigma_{00}=\begin{bmatrix}
1&0\\0&1\end{bmatrix}, \,\,\,\text{and}\,\,\,\Sigma_{01}=\begin{bmatrix}
3&1\\1&3\end{bmatrix}.
\]
As expected, the fixed-point algorithm converges for both of the initial values. In both cases, convergence occurred after five steps and gives the estimate:
\[
\hat{\Sigma}_{\eta}=\begin{bmatrix}
0.5136 & -0.0072\\
-0.0072 & 0.5038
\end{bmatrix}.
\]

Next, consider the Riemannian penalty $\Pi_R(\Sigma)=\Vert\log\Sigma\Vert^2_F$, for which the fixed-point algorithm is given by:
\[
\Sigma_{k+1}\leftarrow \frac{\nu+p}{n}\sum_{i=1}^n\frac{x_ix_i^{\tran}}{\nu+x_i^{\tran}\Sigma^{-1}_kx_i}-2\eta\log(\Sigma_k)\Sigma_k.
\]
For $\eta=0.5$, we use this algorithm with the same data and the same two initial values as before.
\begin{itemize}
\item[(i)] With $\Sigma_{00}$, it converges at the $38th$ step and gives an estimate $$\hsig_{\eta}=\begin{bmatrix}
0.9758 & -0.0063\\ -0.0063 & 1.1291\end{bmatrix}.$$
\item[(ii)] With $\Sigma_{01}$, it fails to converge.
\end{itemize}

\end{example}

In general, the minima of M-optimization problem \eqref{eq:mL} do not have closed form solutions. We show
here, though, that if one knows how to minimize \eqref{eq:nL}, the normal optimization problem, then one can use a simple reweighting algorithm
to find a minimum to \eqref{eq:mL}, whenever $\rho(s)$ is concave as well as g-convex.
For such cases, the weight function $u(s) = \rho^\prime(s)$ is non-increasing, while the ``influence function'' $\psi(s) = su(s)$ is
non-decreasing. These are the original conditions imposed on the weight and influence functions for the multivariate M-estimators \cite{Maronna:1976}. Examples which satisfies these conditions include the $u(s)$ and $\psi(s)$ functions associated with the sample covariance matrix, the maximum likelihood estimators of scatter under an elliptical $t$-distribution on $\nu$ degrees of freedom, and Huber's M-estimator of scatter.

A first order approximation to $\ell_\rho(\Sigma)$ yields the following monotonic algorithm.
\begin{theorem} \label{Thrm:alg} Suppose $\rho(s)$ is concave and differentiable with $u(s) = \rho^\prime(s)$,
and let $L(\Sigma;M;\eta)$ being defined as in \eqref{eq:nL}. Given an initial $\Sigma_0 > 0$, define the
sequence 
\[ \Sigma_{k+1} = \argmin_{\Sigma > 0} L(\Sigma; M(\Sigma_k);\eta), \ \mbox{where} \
M(\Sigma) = \frac{1}{n}\sum_{i=1}^n u(x_i^{\tran} \Sigma^{-1} x_i) x_i x_i^{\tran}. \]  
Then $L_{\rho}(\Sigma_{k+1};\eta) \le L_{\rho}(\Sigma_k;\eta)$,  i.e.\ $L_{\rho}(\Sigma_k;\eta)$ is a non-increasing sequence.
\end{theorem} 

For the unpenalized case, $\Sigma_{k+1} = M(\Sigma_k)$, and so this reweighting algorithm coincides with the original fixed-point algorithm proposed in \cite{Maronna:1976} and \cite{Huber:1977}. When using the Kullback-Leibler penalty or the trace precision matrix penalty this algorithm can be shown to be equivalent to the FP algorithm. 
The algorithm can also be used with non-smooth penalties. In particular, it readily applies to the non-smooth elasso penalty $\Pi_E(\Sigma)$. This application is facilitated by the finite-step algorithm presented in \cite{Tyler-Yi:2018}, which can efficiently find the minimum of the normal objective $L(\Sigma; M(\Sigma_k);\eta)$. Note that a special case of the elasso penalty is the log-condition number penalty $\Pi_{cn}(\Sigma)=\log(\lambda_1)-\log(\lambda_p)$.

To establish the convergence of the sequence $\Sigma_k$, the solution set $\mA _{\rho}$ needs to be non-empty and compact.
We then have the following convergence theorem.

\begin{theorem} \label{Thrm:converge} Let $\rho(s)$ be concave and g-convex, and $\Pi(\Sigma)$ be g-convex. 
If the solution set $\mA_{\rho}$ as defined by \eqref{eq:sol} is non-empty and compact, then
the sequence $\Sigma_k$ has an accumulation point, with any accumulation point, say $\Sigma_{\infty}$, in $\mA_{\rho}$.
Furthermore, if either $\rho(s)$ or $\Pi(\Sigma)$ is strictly g-convex, then $\mA_{\rho}$ contains
exactly one element, say $\Sigma_{\infty}$, with $\Sigma_k \to \Sigma_{\infty}$.
\end{theorem}

\begin{example}\label{ex:2}
To illustrate our new algorithm, consider the same data set used in Example \ref{ex:1}. Recall that when using the $\Pi_o(\Sigma)$ penalty, our algorithm corresponds to the fixed-point algorithm. Unlike the Kullback-Liebler and $\Pi_o(\Sigma)$ penalty, the normal minimization problem $\Sigma_{k+1} = \argmin_{\Sigma > 0} L(\Sigma; M(\Sigma_k);\eta)$ generally does not have a closed-form solution for other penalty functions. For the Riemannian penalty $\Pi_R(\Sigma)$, however, a simple iterative algorithm is provided in \cite{Tyler-Yi:2019}. Applying our proposed algorithm with $\eta=0.5$ and the same two initials values from Example 1, we observe convergence for both $\Sigma_{00}$ and $\Sigma_{01}$. The algorithm converges at the $12th$ and $14th$ iterations, respectively, reaching the value of $\hsig_{\eta}$ given in Example \ref{ex:1}.
\end{example}

\begin{itemize}
    \item \emph{Remark 1.} Our algorithm can also be interpreted within the majorization-minimization (MM) framework, where $L(\Sigma;M(\Sigma_k);\eta)$, the normal objective function, serves as a surrogate function that locally approximates the M-objective function $L_{\rho}(\Sigma;\eta)$ with their difference minimized at the current point. In other words, the M-objective function $L_{\rho}(\Sigma;\eta)$ is upper-bounded by the surrogate function up to a constant, that is, $L_{\rho}(\Sigma;\eta)\le L(\Sigma;M(\Sigma_k);\eta)+\frac{1}{n}\sum_{i=1}^n\rho(x_i^{\tran}M^{-1}(\Sigma_k)x_i)$. However, unlike traditional MM approaches such as those in \cite{Sun:2014} and \cite{Wiesel:2012}, which require deriving problem-specific surrogate functions for each combination of likelihood and penalty--limiting their applicability to Tyler's loss and Kullback-Leibler type penalties--our re-weighting algorithm provides a unified framework. We establish global convergence for a broad class of loss functions $\rho$ and penalty functions $\Pi(\Sigma)$, without requiring case-by-case surrogate function construction. This generality enables practitioners to apply our method to diverse robust estimation problems without redeveloping the theoretical machinery for each new objective function.
\item \emph{Remark 2.}  The conditions in the above two theorems apply to the Kullback-Leibler and the trace precision matrix penalties. As already noted, our proposed re-weighting algorithm is equivalent to the FP algorithm when using either of these two penalties.  Theorem \ref{Thrm:alg} then implies the FP algorithm given in \cite{Ollila-Tyler:2014} is monotonic. It was also noted that only a partial proof is given in \cite{Ollila-Tyler:2014} for the convergence of their FP algorithm. Theorem \ref{Thrm:converge} then provides a complete proof. 
\end{itemize}

\section{Constrained optimization}\label{Sec:constraint}
It has been noted that in some applications, covariance matrices may possess certain structures, such as Banded matrices, which are associated with time-varying moving average models, as studied in \cite{Bickel-Levina:2008}, or Toeplitz matrices, which are used to model the correlation of cyclostationary processes in periodic time series, see \cite{Miller-Snyder:1987}. In such cases, one may wish to consider a structured maximum likelihood estimator or a structured M-estimator of scatter, defined as:
\begin{equation}\label{eq:cL}
    \min_{\Sigma > 0} \ell_{\rho}(\Sigma), \quad \text{s.t.} \quad \Sigma\in\MM,
\end{equation}
where $\MM\subset\mP^p$ is a constraint set that characterizes the covariance structure. When the constraint set corresponds to Banded or Toeplitz matrix classes, these form convex subsets of the positive semidefinite cone. Consequently, if $\ell_\rho(\Sigma)$ is the Gaussian negative log-likelihood, standard convex optimization methods can be used to solve the constrained problem \eqref{eq:cL}. For Tyler's likelihood under structural constraints, \cite{Sun:2016} applied MM algorithm with convex surrogate functions but provided no convergence guarantees--establishing neither local nor global convergence. Furthermore, they did not address well-posedness, leaving unresolved when solutions exist and are unique. Beyond convex structures, important classes of covariance matrices such as Kronecker matrices \cite{Srivastava:2008} lack standard convexity. Additionally, as already noted $\ell_\rho(\Sigma)$ is not convex in both $\Sigma$ and $\Sigma^{-1}$ for most M-estimators of scatter of interest. However, the class of Kronecker matrices is geodesic convex. We therefore establish well-posedness for the constrained problem and extend our penalized algorithm to handle constraints whenever $\MM$ forms a geodesically convex set. 

Typically, constrained optimization problems \eqref{eq:cL} do not have closed form solutions and need to be solved using iterative numerical techniques. According to optimization theory, if both the objective function and the constraint set are convex, then due to the Lagrange duality, a constrained optimization problem is mathematically equivalent to a penalized optimization problem. Similarly, we note that under g-convexity, results for the penalized likelihood or M-estimation framework \eqref{eq:mL} readily apply to the constrained likelihood or M-estimation framework \eqref{eq:cL} due to the following duality between the two settings. 

\begin{theorem}\label{thrm:duality}
	Suppose $L_{\rho}(\Sigma;\eta)$ has a unique minimizer $\hat{\Sigma}_{\eta}>0$ which is continuous as a function of $\eta$. Define $\kappa_L=\inf\{\Pi(\Sigma) \,|\, \Sigma>0\}$ and $\kappa_U=\sup\{\Pi(\Sigma)\,|\,\eta>\eta_o\}$, where $\eta_o$ satisfied the conditions for Theorem \ref{thrm:largeeta0} or Theorem \ref{thrm:largeeta}. For $\kappa_L<\kappa\le\kappa_U$, there exists a unique solution $\tilde{\Sigma}_{\kappa}>0$ to the problem $\arg\min\{\ell_{\rho}(\Sigma) \,|\, \Sigma>0, \Pi(\Sigma)\le\kappa\}$, with the solution $\tilde{\Sigma}_{\kappa}$ being a continuous function of $\kappa$. Furthermore, for each $\eta>\eta_o$ there exists a $\kappa>0$, and vice versa, such that $\hat{\Sigma}_{\eta}=\tilde{\Sigma}_{\kappa}$. The relationship between $\eta$ and $\kappa$ is given by $\kappa(\eta)=\Pi(\hat{\Sigma}_{\eta})$ for $\eta > \eta_o$.
\end{theorem}

Therefore, if the g-convex constraint set can be expressed using a g-convex function, the existence of a solution to \eqref{eq:cL} follows immediately from the penalized problem \eqref{eq:mL}. For example, for the condition number constraint problem first studied by \cite{Won:2013}, the constraint set $\{\Sigma\in\mP^p\,|\, \lambda_1/\lambda_p\le\kappa\}$ can be related to the log-condition number penalty $\Pi_{cn}$.

Theoretically, one can compute the solution to \eqref{eq:cL} for a given $\kappa$ by first computing the solution to \eqref{eq:mL} for a range of $\eta$ and then
finding the value of $\eta$ which gives $\Pi(\hat{\Sigma}_{\eta}) = \kappa$. This is unnecessarily complicated since we can use the iterative algorithm \eqref{eq:calg} proposed below to solve \eqref{eq:cL} directly. The proof of the convergence of \eqref{eq:calg} makes use of the duality between \eqref{eq:cL} and  \eqref{eq:mL}. We first give a lemma that characterizes the solution set to the constraint problem (\ref{eq:cL}).
\begin{lemma}\label{Lem:sol-set}
Suppose $\MM\subset\mP^p$ is a closed g-convex set, and $\ell_{\rho}(\Sigma)$ is g-coercive, then the solution set 
\begin{equation}\label{eq:csolutionset}
\mA^*_{\rho}=\{\Sigma^*\in\MM | \ell_{\rho}(\Sigma^*)\le \ell_{\rho}(\Sigma), \Sigma\in\MM\}
\end{equation}
is not empty and $\mA^*_{\rho}$ is g-convex.
\end{lemma}

For a general closed g-convex constraint set, similar to the algorithm for the penalized problem, it is presumed that the solution to the constrained likelihood problem under the normal negative log-likelhood function $\ell(\Sigma;S_n)$ can be obtained. For some initial $\Sigma_0 \in \MM$, the $(k+1)$-th step the update of $\Sigma$ is then given by 
\begin{equation}\label{eq:calg}
\Sigma_{k+1}=\argmin_{\Sigma\in\MM} \ell(\Sigma;M(\Sigma_k)), 
\end{equation}
where $M(\Sigma)$ is the weighted covariance matrix defined in Theorem \ref{Thrm:alg}.
 We show below that at each step of this reweighting algorithm the objective function $\ell_{\rho}$ decreases.

\begin{lemma} \label{lem:calg} 
Suppose $\rho(s)$ is strictly concave and differentiable with $u(s) = \rho^\prime(s)$. Given an initial $\Sigma_0 > 0\in\MM$, define the
sequence $\Sigma_{k+1}$ as in \eqref{eq:calg}. 
Then $\ell_{\rho}(\Sigma_{k+1}) < \ell_{\rho}(\Sigma_k)$,  i.e.\ $\ell_{\rho}(\Sigma_k)$ is a monotone decreasing sequence, unless $\Sigma_{k+1}=\Sigma_k$.
\end{lemma} 

Finally, using these two previous lemmas, we can establish the convergence of the iterative reweighting algorithm \eqref{eq:calg}.
\begin{theorem} \label{thrm:cconverge} 
Let $\rho(s)$ be strictly concave and g-convex. 
If the solution set $\mA_{\rho}^*$ as defined by \eqref{eq:csolutionset} is non-empty and compact, then the sequence $\Sigma_k$, as defined by \eqref{eq:calg}, has an accumulation point, with any accumulation point, say $\Sigma_{\infty}$, in $\mA_{\rho}^*$.
Furthermore, if $\rho(s)$ is strictly g-convex, then $\mA_\rho^*$ contains
exactly one element, say $\Sigma_{\infty}$, with $\Sigma_k \to \Sigma_{\infty}$.
\end{theorem}

\begin{example}
For matrix-valued data, say $X_i\in\R^{p_1\times p_2}, i=1,\cdots,n$, it is common to model its covariance structure by using a Kronecker product model; that is $\Sigma=\Sigma_1\otimes \Sigma_2$, where $\Sigma_1 \in \mathcal{P}^{p_1}$ and $\Sigma_2 \in \mathcal{P}^{p_2}$ are proportional to the covariance matrices of the row and column variables, respectively. Further restrictions like group symmetry on the matrices $\Sigma_1$ and/or $\Sigma_2$ may also be imposed. Let $\mathcal{K}_j$ be a set of orthogonal matrices on $\R^{p_j}, j = 1,2$. Consider the set
\begin{equation*}
\mathcal{M}_2=\{(\Sigma_1, \Sigma_2)\in \mathcal{P}^{p_1}\otimes\mathcal{P}^{p_2}\,\,|\,\, \det(\Sigma_j)=1, U_j\Sigma_jU_j^{\tran}=\Sigma_j, \forall U_j\in\mathcal{K}_j, j=1,2\}.
\end{equation*}
Following \cite{Wiesel:2012b}, it can be shown that $\mathcal{M}_2$ is a g-convex set. So, applying our constrained re-weighting algorithm, an M-estimator of scatter having both the kronecker and group symmetry structure can be computed via
\begin{equation*}
\scriptstyle
\begin{cases}
\tilde{\Sigma}_{1, k+1}=\frac{1}{np_2|\mathcal{K}_1||\mathcal{K}_2|}\sum_{\substack{U_j\in\mathcal{K}_j,\\  j=1,2}}\sum_{i=1}^nu\{\tr(\Sigma_{1,k}^{-1}Y_i\Sigma_{2,k}^{-1}Y_i^{\tran})\}Y_i\Sigma_{2,k}^{-1}Y_i^{\tran}\\
\tilde{\Sigma}_{2, k+1}=\frac{1}{np_1|\mathcal{K}_1||\mathcal{K}_2|}\sum_{\substack{U_j\in\mathcal{K}_j,\\ j=1,2}}\sum_{i=1}^nu\{\tr(\Sigma_{1,k}^{-1}Y_i\Sigma_{2,k}^{-1}Y_i^{\tran})\}Y_i\Sigma_{1,k}^{-1}Y_i^{\tran}\\
\Sigma_{1,k+1}=\frac{\tilde{\Sigma}_{1,k+1}}{(\det(\tilde{\Sigma}_{1,k+1}))^{1/p_1}}\\
\Sigma_{2,k+1}=\frac{\tilde{\Sigma}_{2,k+1}}{(\det(\tilde{\Sigma}_{2,k+1}))^{1/p_2}}\,,
\end{cases}
\end{equation*}
where $Y_i=U_1X_iU_2$.
\end{example}

\begin{itemize}
    \item  \emph{Remark 3.} Note that if we only assume the covariance has the group symmetry structure, then our algorithm reduces to the one developed in \cite{Soloveychik:2015}. If we only assume the Kronecker structure, use Tyler's weight function, and change the identification condition $\det(\Sigma_j) = 1$ to $\Vert \Sigma_j \Vert_2=1$, our algorithm becomes a generalization of the algorithm considered in \cite{Soloveychik-Trushin:2015}.
    \item \emph{Remark 4.} Example 3 can be readily generalized to K-mode multi-way array data $X_i\in\R^{p_1\times\cdots\times p_K}, i=1,\cdots,n$, whose covariance has the tensor structure $\Sigma=\Sigma_1\otimes\cdots\otimes\Sigma_K$, where $\Sigma_k \in \mathcal{P}^{p_k}, k=1,\cdots,K$. It can again be shown that the set $\mathcal{M}_K$ is g-convex, and so our constrained re-weighting algorithm can be applied to construct M-estimators of scatter for tensor covariance models. 
\end{itemize}

\section{Concluding remarks}\label{Sec:con}
There has been considerable interest, especially in the signal processing and the EE community, in penalized and structured M-estimation of multivariate scatter. Much of the previous literature addresses more specific problems,  with the corresponding arguments not always being complete. In this paper, we have established a unified theoretical and computational framework for penalized and structured M-estimators of multivariate scatter. By leveraging geodesic convexity, we provide comprehensive well-posedness results that apply to broad classes of penalty functions and structural constraints, moving beyond the case-by-case analyses prevalent in existing literature.

The theoretical analysis resolves several open questions in robust covariance estimation. Existence and uniqueness conditions are established for penalized M-estimators under geodesically convex penalties, including non-smooth cases previously unaddressed. For structured estimation, the duality relationship between penalized and constrained formulations immediately yields existence results for geodesically convex constraint sets.

Our re-weighting algorithm addresses fundamental convergence failures of the standard fixed-point iteration for penalty functions beyond the Kullback-Leibler type penalties. The algorithm achieves monotone convergence for general geodesically convex penalties while maintaining the same convergence rate as fixed-point methods when they do converge. Critically, the algorithm reduces each iteration to solving a penalized Gaussian problem—a well-studied subproblem with efficient solutions—thereby providing practitioners with an implementable method that requires minimal problem-specific derivation. The extension to constrained problems demonstrates that the same algorithmic framework applies to geodesically convex constraint sets, unifying the treatment of penalized and structured estimation. This generality encompasses non-convex but geodesically convex structures, including covariance structures like the Kronecker products, group-symmetric or log-condition number, providing convergence guarantees absent in previous work.


\section*{Acknowledgments}
This work was supported by the National Science Foundation Grants DMS-1407751 and DMS-1812198, the National Natural Science Foundation of China (No. 12101119) and the Fundamental Research Funds for the Central Universities.
\appendix
\renewcommand{\thesection}{Appendix \Alph{section}}
\titleformat{\section}
  {\normalfont\normalsize\bfseries}{\thesection.}{0.5em}{}

\section{Geodesic convexity} \label{App:gc}
The set of symmetric positive definite matrices of order $p$ can be viewed as a Riemannian manifold with the geodesic path from $\Sigma_0 > 0$ to 
$\Sigma_1 > 0$ being given by $\Sigma_t = \Sigma_0^{1/2} \{\Sigma_0^{-1/2}\Sigma_1\Sigma_0^{-1/2}\}^t \Sigma_0^{1/2}$ for \mbox{$0 \le t \le 1$}, see \cite{Wiesel:2012, Duembgen-Tyler:2016}
for more details. Any alternative representation for this path is given by \mbox{$\Sigma_t = Be^{t\Delta_1} B^{\tran}$}, where $\Sigma_0 = BB^{\tran}$ 
and $\Sigma_1 = Be^\Delta_1 B^{\tran}$
with $\Delta$ being a diagonal matrix of order $p$. A function $f(\Sigma)$ is said to be geodesically convex, or g-convex, if and only if
$f(\Sigma_t) \le (1-t)f(\Sigma_0) + tf(\Sigma_1)$ for $0 < t < 1$, and it is strictly g-convex if strict inequality holds.

As previously noted, g-convex functions possess properties similar to those of convex functions. In particular,
a g-convex function define on $\Sigma > 0$ is continuous, and any critical point of a g-convex function corresponds to a global minimum. 
When the g-convex function is differentiable, then $\Sigma_0 = BB^{\tran}$ is a global minimum if and only if 
\begin{equation} \label{eq:dirder}
 \lim_{t \to 0^+} \frac{f(Be^{t\Delta}B^{\tran}) - f(BB^{\tran}) }{t} \ge 0, \ \mbox{for any diagonal matrix} \ \Delta \ \mbox{of order} \ p,
\end{equation}
see e.g. Lemma 3.11 in \cite{Duembgen-Tyler:2016}. Furthermore, the set of all minima of a g-convex function form a 
non-empty compact set if and only if the function is also g-coercive. A function $f(\Sigma)$ defined for $\Sigma > 0$ is said
to be g-coercive if $f(\Sigma) \to \infty$ as $\| \det\{\log(\Sigma)\}  \| \to \infty$. More generally, for a g-convex and g-coercive function, the level sets
$\mathcal{C}_c = \{ \Sigma > 0 \ | \ f(\Sigma) \le c \}$ are non-empty compact g-convex sets for any $c \ge \inf_{\Sigma> 0} f(\Sigma)$, where a
subset $\mathcal{C}$ of symmetric positive definite matrices is said to be g-convex if $\Sigma_0$ and $\Sigma_1 \in \mathcal{C}$ implies $\Sigma_t \in \mathcal{C}$.

\section{Proofs} \label{App:proofs} 

\subsection*{\small \emph{Proof of Lemma \ref{Lem:sumgeo}}}
As previously noted, this lemma follows from properties of g-convex function. In particular, it has been shown
by \cite{Zhang:2013} that when $\rho(s)$ is g-convex, then $\ell_\rho(\Sigma)$ is a g-convex function of $\Sigma >0$,
and is strictly g-convex when $\rho(s)$ is strictly g-convex and at least one $x_i \ne 0$. Hence, the g-convexity
properties of $L_\rho(\Sigma; \eta)$ follows. \quad $\square$
\\

\subsection*{\small\emph{Proof of Lemma \ref{Lem:Lexist}}}
The first part of the lemma follows from noting that for any fixed $\Sigma$, the term $L_\rho(\Sigma; \eta)$ is non-decreasing in $\eta$.
Hence, if $L_\rho(\Sigma; \eta)$ is g-coercive at $\eta = \eta_o$, then it is g-coercive at $\eta > \eta_o$. The result then follows
since $\mathcal{A}_\rho$ is a non-empty concave set if and only if $L_\rho(\Sigma; \eta)$ is g-coercive.

To prove continuity, we first state the following general lemma.
\begin{lemma} \label{Lem:cont}
Let $\mathcal{D}$ be a closed subset of $\R^p$. Suppose the real-valued functions $f(x)$ and $g(x)$ are continuous on $\mathcal{D}$, with $g(x) > 0$.
Furthermore, suppose $h(x;\eta) = f(x) + \eta g(x)$ has a unique minimum in $\mathcal{D}$ for any $0 \le \eta_o \le \eta \le \eta_1$.  If the set 
$\{x \in \mathcal{D} \ | \ h(x;\eta_o) \le c \}$ is compact for any $c \ge \inf \{h(x;\eta_o) \ | \ x \in \mathcal{D} \}$, then the function
$x(\eta) = \mbox{arginf}\{h(x;\eta) \ | \ x \in \mathcal{D} \}$ is continuous for $\eta_o \le \eta < \eta_1$.
\end{lemma}

To prove this lemma, first note that  $h(x;\eta)$ is increasing in $\eta$, and so the set $\{x(\eta) \ | \  \eta_o \le \eta < \eta_1 \}$ is
contained in the compact set $\{ x \ | \ h(x;\eta_o) \le h(x(\eta_1);\eta_1) \}$. So, if $\eta_k \to \eta$, then
$x(\eta_k)$ has a convergent subsequence, say  $x(\eta_{k'}) \to \tilde{x}$. By definition, 
$h(x(\eta_{k'});\eta_{k'}) \le h(x(\eta);\eta_{k'})$. By continuity, the left hand side converges to 
$h(\tilde{x};\eta)$ and the right hand side converges to $h(x(\eta);\eta)$. By uniqueness, this implies $\tilde{x} = x(\eta)$. Hence,
$x(\eta_k) \to x(\eta)$, which establishes Lemma \ref{Lem:cont}. 

The continuity assertion follows from Lemma \ref{Lem:cont}, with $f(\Sigma) = \ell_\rho(\Sigma)$ an $g(\Sigma) = \Pi(\Sigma)$,
since, by g-convexity, $\ell_\rho(\Sigma)$ and $\Pi(\Sigma)$ are continuous, and by g-coercivity, the level sets of $L_\rho(\Sigma; \eta)$
are compact. \quad $\square$ \\

\subsection*{\small\emph{Proof of Theorem \ref{thrm:largeeta0}} }
We only need to show that $L_{\rho}(\Sigma; \eta)$ is g-coercive for some $\eta>0$. Since $\rho(s)$ is bounded below, it is sufficient to show
\[
D(\Sigma; \eta)=\log\{\det(\Sigma)\}+\eta\Pi(\Sigma)=\Pi(\Sigma)\{\eta+\log\{\det(\Sigma)\}/\Pi(\Sigma)\}
\]
is g-coercive. Consider a sequence $\Sigma_k>0$ such that $||\log\Sigma_k||\to\infty$, which implies $\Pi(\Sigma_k)\to\infty$ since $\Pi(\Sigma)$ is g-coercive. If $\log\{\det(\Sigma_k)\}$ is bounded below, then it readily follows that $D(\Sigma_k; \eta)\to\infty$ for any $\eta>0$. Also, if $\log\{\det(\Sigma_k)\}\to -\infty$ then $\log\{\det(\Sigma_k)\}/\Pi(\Sigma_k)>-\eta_o$ for large enough $k$, and so $D(\Sigma_k; \eta)\to\infty$ for $\eta>\eta_o$.\quad $\square$
\\

\subsection*{\small\emph{Proof of Theorem \ref{thrm:largeeta}} }
As already noted, this theorem is established within the proof of Theorem \ref{thrm:largeeta0}.
  \quad $\square$ \\

\subsection*{\small\emph{Proof of Theorem \ref{thrm:generalexist}} }
Since $x_i^{\tran}\Sigma^{-1}x_i\ge s_i/\lambda_1$, where $s_i=x_i^{\tran}x_i$ and $\rho(s)$ is non-decreasing, it follows that $L_{\rho}(\Sigma; \eta)\ge H_{\rho}(\lambda_1)+h(\Gamma; \eta)$, where $\Gamma=\Sigma/\lambda_1$,
\[
H_{\rho}(\lambda_1)=\frac{1}{n}\sum_{i=1}^n\{\rho(s_i/\lambda_1)+p\log\lambda_1\} \quad \text{ and }\quad h(\Gamma; \eta)=\log\{\det(\Gamma)\}+\eta\Pi(\Gamma),
\]
Now, $||\log\Sigma||_F\to\infty$ if and only if $\lambda_1\to\infty$ or $\gamma_p\to 0$. Since $\gamma_1=1$, it follows by the conditions of the theorem that if $\gamma_p\to 0$ then $h(\Gamma; \eta)\to\infty$ for $\eta\ge\eta_o$. If $\lambda_1\to\infty$, then $H_{\rho}(\lambda_1)\to\infty$, whereas $\lambda_1\to 0$, then by the definition of the sill $K_{\rho}$ and the condition on $P_n(0), H_{\rho}(\lambda_1)\to\infty$. Thus, if $\lambda_1\to\infty$ and/or $\gamma_q\to 0$ then $L_{\rho}(\Sigma; \eta)\to\infty$ for $\eta\ge\eta_o$. Hence $L_{\rho}(\Sigma; \eta)$ is g-coercive for such $\eta$, and the theorem then follows. \quad $\square$\\

\subsection*{\small\emph{Proof of Theorem \ref{thrm:tylerexist}} }
Since $x_i^{\tran}\Sigma^{-1}x_i\ge x_i^{\tran}x_i/\lambda_1$, it follows that
\[
L_{\rho}(\Sigma; \eta)\ge \frac{p}{n}\sum_{i=1}^n\log(x_i^{\tran}x_i)+\log\{\det(\Gamma)\}+\eta\Pi(\Sigma)
\]
where $\Gamma=\Sigma/\lambda_1$. Then the result follows by the similar proof of Theorem \ref{thrm:largeeta0}.

\subsection*{\small\emph{Proof of Theorem \ref{Thrm:alg}}}
Since $\rho(s)$ is concave, $\rho(s) \le \rho(s_o) +(s-s_o)u(s_o)$, and so
$\ave\{\rho(x_i^{\tran} \Sigma_{k+1}^{-1} x_i)\} \le \ave\{\rho(x_i^{\tran}\Sigma_{k}^{-1} x_i)\} + \tr\{(\Sigma_{k+1}^{-1} - \Sigma_k^{-1})M(\Sigma_k)\}.$ Furthermore, by definition $L(\Sigma_{k+1}; M(\Sigma_k); \eta) \le L(\Sigma_{k}; M(\Sigma_k); \eta)$. Together, these two inequalities imply 
\begin{align}\label{eq:Lle}
L_{\rho}(\Sigma_{k+1}; \eta)  & = L(\Sigma_{k+1}; M(\Sigma_k); \eta) + \ave\{\rho(x_i^{\tran} \Sigma_{k+1}^{-1} x_i)\} - \tr\{\Sigma_{k+1}^{-1}M(\Sigma_k)\}  \nonumber\\
& \le L(\Sigma_{k}; M(\Sigma_k); \eta) + \ave\{\rho(x_i^{\tran} \Sigma_{k}^{-1} x_i)\} - \tr\{\Sigma_k^{-1}M(\Sigma_k)\} \\
 & =  L_{\rho}(\Sigma_k; \eta). \nonumber\quad \square
\end{align}
\\

\subsection*{\small\emph{Proof of Theorem \ref{Thrm:converge}}}
We first establish the following lemma.

\begin{lemma} \label{Lem:Paccum}
The sequence $\Sigma_k$ has an accumulation point, with any accumulation point, say $\Sigma_{\infty}$, satisfying 
$L(\Sigma_{\infty}; M_{\infty};\eta) = \inf_{\Sigma > 0} L(\Sigma; M_{\infty};\eta)$, where $M_{\infty} = M(\Sigma_{\infty})$
\end{lemma}

To show this, recall that the solution set $\mA_{\rho}$ is non-empty and compact if and only if the function $L_{\rho}(\Sigma;\eta)$ is g-coercive.
Consequently the set $\mathcal{S}_o = \{\Sigma > 0 ~|~ L_\rho(\Sigma;\eta) \le L_\rho(\Sigma_0;\eta) \} $ is g-convex and compact, with
the sequence $\Sigma_k \in \mathcal{S}_o$. Hence, there exists a convergent sub-sequence $\Sigma_{k'}$ such that
$\Sigma_{k'} \to \Sigma_{\infty}$ and $\Sigma_{k'+1} \to \Sigma_{1,\infty}$, where $\Sigma_{1,\infty} = \argmin_{\Sigma > 0} L(\Sigma; M_\infty; \eta)$.
Since $L_{\rho}(\Sigma_k; \eta)$ is non-increasing, it follows that $L_{\rho}(\Sigma_\infty; \eta) = L_{\rho}(\Sigma_{1,\infty}; \eta)$. This implies
$L(\Sigma_{1,\infty}; M_{\infty}; \eta) = L(\Sigma_{\infty}; M_{\infty}; \eta)$, since otherwise \eqref{eq:Lle}, with $\Sigma_k$ set to $\Sigma_\infty$,
would be a strict inequality, which contradicts $L_\rho(\Sigma_\infty; \eta) = L_\rho(\Sigma_{1,\infty}; \eta)$. Hence Lemma \ref{Lem:Paccum} holds. 

Continuing, since $\Sigma_\infty$ is a minimizer of the g-convex function $L(\Sigma; M_\infty; \eta)$, it follows from \eqref{eq:dirder} that
$\lim_{t \to 0^+} \{L(Be^{tA}B^{\tran}; M_\infty; \eta)- L(BB^{\tran}; M_\infty; \eta)\}/t \ge 0$ for any $A$ and
for $BB^{\tran} = \Sigma_\infty$. Next, express $L_\rho(\Sigma; \eta) = L(\Sigma; M_\infty; \eta) + Q(\Sigma)$,
where $Q(\Sigma) = n^{-1}\sum_{i=1}^n \{\rho(x_i^{\tran} \Sigma^{-1} x_i)\} - \tr(\Sigma^{-1}M_\infty)$. 
The function $Q$ is differentiable, and it can be readily shown that $\lim_{t \to 0^+} \{Q(Be^{tA}B^{\tran}) - Q(BB^{\tran}) \}/t = 0$. 
Thus, $\lim_{t \to 0^+}$ $ \{L_\rho(Be^{tA}B^{\tran}; \eta)  - L_\rho(BB^{\tran}; \eta) \}/t \ge 0$ for any $A$, and hence 
$\Sigma_\infty$ is a minimizer of the g-convex function $L_\rho(\Sigma; \eta)$. \quad $\square$ \\

\subsection*{\small\emph{Proof of Theorem \ref{thrm:duality}}}
First note that $\kappa(\eta)$ is continuous and non-increasing with $\kappa_L\le\kappa(\eta)\le\kappa_U$. Continuity follows since both $\Pi$ and $\hsig_{\eta}$ are continuous. To prove that $\kappa_{\eta}$ is non-increasing, suppose $\eta_1<\eta_2$ and define, for $j=1, 2, \hsig_j=\hsig_{\eta_j}, \ell_{\rho,j}=\ell_{\rho}(\hsig_j)$ and $\kappa_j=\kappa_{\eta_j}$. By definition of $\hsig_j$, it follows that $\ell_{\rho,1}+\eta_1\kappa_1\le \ell_{\rho,2}+\eta_1\kappa_2$ and $\ell_{\rho,2}+\eta_2\kappa_2\le \ell_{\rho,1}+\eta_2\kappa_1$, which together implies that $\eta_1(\kappa_1-\kappa_2)\le \ell_{\rho,2}-\ell_{\rho,1}\le\eta_2(\kappa_1-\kappa_2)$. Hence $\kappa_1\ge\kappa_2$. Furthermore, if $\kappa_1=\kappa_2$, then $\ell_{\rho,1}=\ell_{\rho,2}$ and $\ell_{\rho,1}+\eta_1\kappa_1= \ell_{\rho,2}+\eta_1\kappa_2$, which, by uniqueness, implies $\hsig_1=\hsig_2$.

Next, for $\kappa_L<\kappa\le\kappa_U$, define $\eta(\kappa)=\inf\{\eta\ge 0\,|\, \kappa(\eta)=\kappa\}$, and note that $\Pi(\hsig_{\eta(\kappa)})=\kappa$. It readily follows from the previous paragraph that $\eta(\kappa)$ is strictly decreasing and continuous from above. By definition, $\ell_{\rho}(\hsig_{\eta(\kappa)})+\eta(\kappa)\kappa\le \ell_{\rho}(\Sigma)+\eta(\kappa)\Pi(\Sigma)$ for any $\Sigma>0$, and so $\ell_{\rho}(\hsig_{\eta(\kappa)})-\ell_{\rho}(\Sigma)\le\eta(\kappa)(\Pi(\Sigma)-\kappa)$. Hence, if $\Pi(\Sigma)\le\kappa$, then $\ell_{\rho}(\hsig_{\eta(\kappa)})\le \ell_{\rho}(\Sigma)$, which implies $\tilde{\Sigma}_{\kappa}=\hsig_{\eta(\kappa)}$. Since $\hsig_{\eta}$ is continuous, it follows that $\tilde{\Sigma}_{\kappa}$ is continuous at points of continuity of $\eta(\kappa)$. It is also continuous at points of discontinuity of $\eta(\kappa)$ since $\hsig_{\eta}$ is constant on the sets $\{\eta\,|\, \kappa(\eta)=\kappa\}$. \quad $\square$
\\

\subsection*{\small\emph{Proof of Lemma \ref{Lem:sol-set}}}
If $\ell_{\rho}(\Sigma)=\infty$ for every $\Sigma\in\MM$, then every point is a global minimizer. WLOG, assume there exists a $\Sigma_*$ such that $\ell_{\rho}(\Sigma_0)<\infty$. In this case, it follows that $\ell^*_{\rho}:=\inf_{\Sigma\in\MM}\ell_{\rho}(\Sigma)\le \ell_{\rho}(\Sigma_0)<\infty$.

Now let $\Sigma_k$ be the minimizing sequence of the optimization problem. By the coercivity, $\Sigma_k$ is bounded, thus admits a convergent subsequence $\Sigma_{k_j}\in\MM\to\Sigma_*\in\MM$ since $\MM$ is closed. Therefore
\[
\inf_{\Sigma\in\MM}\ell_{\rho}(\Sigma)=\ell^*_{\rho}=\lim_{k\to\infty}\ell_{\rho}(\Sigma_k)=\lim_{k_j\to\infty}\ell_{\rho}(\Sigma_{k_j})=\lim\inf_{k_j\to\infty}\ell_{\rho}(\Sigma_{k_j})\ge \ell_{\rho}(\Sigma_*)\]
which shows that $\Sigma_*$ is a global minimizer in $\MM$, thus $\mA^*_{\rho}$ is non-empty.

Suppose $\Sigma_1,\Sigma_2\in \MM$ and $\ell_{\rho}(\Sigma_1)\le \ell_{\rho}(\Sigma), \ell_{\rho}(\Sigma_2)\le \ell_{\rho}(\Sigma)$ for all $\Sigma\in\MM$, then by g-convexity of $\MM$, $\Sigma_t\in\MM$, and by g-convexity of $\ell_{\rho}$,
\[
\ell_{\rho}(\Sigma_t)\le(1-t)\ell_{\rho}(\Sigma_1)+t\ell_{\rho}(\Sigma_2)\le \ell_{\rho}(\Sigma), \Sigma\in\MM
\]
Thus $\Sigma_t\in\mA^*_{\rho}$, and $\mA^*_{\rho}$ is g-convex.\quad $\square$ \\

\subsection*{\small\emph{Proof of Lemma \ref{lem:calg}}}
Since $\rho(s)$ is strictly concave, $\rho(s) < \rho(s_o) +(s-s_o)u(s_o)$, and so
$\ave\{\rho(x_i^{\tran} \Sigma_{k+1}^{-1} x_i)\} < \ave\{\rho(x_i^{\tran} \Sigma_{k}^{-1} x_i)\} + \tr\{(\Sigma_{k+1}^{-1} - \Sigma_k^{-1})M(\Sigma_k)\}. $
Furthermore, by definition $\ell(\Sigma_{k+1}; M(\Sigma_k)) \le \ell(\Sigma_{k}; M(\Sigma_k))$, where $\Sigma_k, \Sigma_{k+1}\in\MM$. Together, these two inequalities imply 
\begin{align*}
\ell_{\rho}(\Sigma_{k+1})  & = \ell(\Sigma_{k+1}; M(\Sigma_k)) + \ave\{\rho(x_i^{\tran} \Sigma_{k+1}^{-1} x_i)\} - \tr\{\Sigma_{k+1}^{-1}M(\Sigma_k)\}  \\
& < \ell(\Sigma_{k}; M(\Sigma_k)) + \ave\{\rho(x_i^{\tran} \Sigma_{k}^{-1} x_i)\} - \tr\{\Sigma_k^{-1}M(\Sigma_k)\} \\
 & =  \ell_{\rho}(\Sigma_k). \quad \square
\end{align*}
\\

\subsection*{\small\emph{Proof of Theorem \ref{thrm:cconverge}}}
The sequence $\Sigma_k$ has a convergent subsequence $\Sigma_{k_j}\to\Sigma_{\infty}$. By the continuity of $\ell_{\rho}$, we have $\ell_{\rho}(\Sigma_{k_j})\to \ell_{\rho}(\Sigma_{\infty})$ as $j\to\infty$.

First we show that $\ell_{\rho}(\Sigma_k)\to \ell_{\rho}(\Sigma_{\infty})$ as $k\to\infty$. Observe from Lemma \ref{lem:calg} that $l_{\rho}$ is monotonically decreasing on the sequence $\{\Sigma\}_{k=0}^{\infty}$. Hence we must have $\ell_{\rho}(\Sigma_k)-\ell_{\rho}(\Sigma_{\infty})\ge 0$ for all $k$. Now since $\ell_{\rho}(\Sigma_{k_j})\to \ell_{\rho}(\Sigma_{\infty})$ as $j\to\infty$, given $\epsilon>0$, there exists $j\ge j_0$ such that $\ell_{\rho}(\Sigma_{k_j})-\ell_{\rho}(\Sigma_{\infty})<\epsilon$ for all $j\ge j_0$. Hence for $k\ge k_{j_0}$,
\[
\ell_{\rho}(\Sigma_k)-\ell_{\rho}(\Sigma_{\infty})=\ell_{\rho}(\Sigma_k)-\ell_{\rho}(\Sigma_{k_{j_0}})+\ell_{\rho}(\Sigma_{k_{j_0}})-\ell_{\rho}(\Sigma_{\infty})<\epsilon
\]
Thus $\ell_{\rho}(\Sigma_k)\to \ell_{\rho}(\Sigma_{\infty})$.

Now we want to show that $\Sigma_{\infty}$ is a solution. We prove this by contradiction. Suppose that $\Sigma_{\infty}$ is not a solution. Consider the sequence $\{\Sigma_{k_j+1}\}_{j=1}^{\infty}$ such that $\Sigma_{k_j+1}=\argmin \ell(\Sigma;M(\Sigma_{k_j}))$, which admits a convergent subsequence $\Sigma_{(k_j+1)_l}\to\Sigma_0$ as $l\to\infty$. Since the mapping $\Sigma_{k_j}\mapsto\argmin \ell(\Sigma;M(\Sigma_{k_j}))$ is closed, we have $\Sigma_0=\argmin \ell(\Sigma;M(\Sigma_{\infty}))$. On the other hand, $\ell_{\rho}(\Sigma_k)\to \ell_{\rho}(\Sigma_{\infty})$ implies that  we must have $\ell_{\rho}(\Sigma_0)=\ell_{\rho}(\Sigma_{\infty})$ which contradicts the fact that $\ell_{\rho}(\Sigma_0)<\ell_{\rho}(\Sigma_{\infty})$ for $\Sigma_{\infty}\notin\mA^*_{\rho}$.
\quad $\square$ 


\renewcommand{\refname}{\normalsize\textbf{References}}

\footnotesize

\end{document}